\documentclass{llncs}
\usepackage{amsmath}

\begin{document}





\title{Balanced Families of Perfect Hash Functions \\
and Their Applications}

\author{Noga Alon \inst{1} \and Shai Gutner \inst{2}}
\institute{Schools of Mathematics and Computer Science, Tel-Aviv
University, Tel-Aviv, 69978, Israel.
\thanks{Research supported in part by a grant from the Israel
Science Foundation, and by the Hermann Minkowski Minerva Center
for Geometry at Tel Aviv University.} 
\email{noga@math.tau.ac.il.}
\and School of Computer Science, Tel-Aviv University, Tel-Aviv,
69978, Israel. 
\thanks{This paper forms part of a Ph.D. thesis
written by the author under the supervision of
Prof. N. Alon and Prof. Y. Azar in Tel Aviv University.}
\email{gutner@tau.ac.il.}}

\maketitle

\begin{abstract}
The construction of perfect hash functions is a well-studied topic.
In this paper, this concept is generalized with the following
definition. We say that a family of functions from $[n]$ to $[k]$ is
a $\delta$-balanced $(n,k)$-family of perfect hash functions if for
every $S \subseteq [n]$, $|S|=k$, the number of functions that are
1-1 on $S$ is between $T/\delta$ and $\delta T$ for some constant
$T>0$. The standard definition of a family of perfect hash functions
requires that there will be at least one function that is 1-1 on
$S$, for each $S$ of size $k$. In the new notion of balanced
families, we require the number of 1-1
functions to be almost the same (taking $\delta$ to be close to $1$)
for every such $S$. Our main result is that for any constant $\delta
> 1$, a $\delta$-balanced $(n,k)$-family of perfect hash functions
of size $2^{O(k \log \log k)} \log n$ can be constructed in time
$2^{O(k \log \log k)} n \log n$. Using the technique of color-coding
we can apply our explicit constructions to devise approximation
algorithms for various counting problems in graphs. In particular, we
exhibit a
deterministic polynomial time algorithm for approximating both the
number of simple paths of length $k$ and the number of simple cycles
of size $k$ for any $k \leq O(\frac{\log n}{\log \log \log n})$ in a
graph with $n$ vertices. The approximation is up to any fixed
desirable relative error. 

\textbf{Key words:} approximate counting of subgraphs, color-coding, perfect hashing.

\end{abstract}

\section{Introduction}\label{sec:intro}

This paper deals with explicit constructions of balanced families
of perfect hash functions. The topic of perfect hash functions has
been widely studied under the more general framework of
$k$-restriction problems (see, e.g.,
\cite{Alon:2006:ACS},\cite{Koller:1994:CSS}). These problems have
an existential nature of requiring a set of conditions to hold at
least once for any choice of $k$ elements out of the problem
domain. We generalize the definition of perfect hash functions,
and introduce a new, simple, and yet useful notion which we call balanced
families of perfect hash functions. The purpose of our new
definition is to incorporate more structure into the
constructions. Our explicit constructions together with the method
of color-coding from \cite{JACM::AlonYZ1995} are applied for
problems of approximating the number of times that some fixed
subgraph appears within a large graph. We focus on counting simple
paths and simple cycles. Recently, the method of color-coding has
found interesting applications in computational biology
(\cite{SIKS2006},\cite{SI2006},\cite{journals/bmcbi/ShlomiSRS06},\cite{conf/apbc/HuffnerWZ07}),
specifically in detecting signaling pathways within
protein interaction. This problem is formalized using an
undirected edge-weighted graph, where the task is to find a
minimum weight path of length $k$. The application of our results
in this case is for approximating deterministically the number of
minimum weight paths of length $k$.

\textbf{Perfect Hash Functions.} An $(n,k)$-family of perfect hash
functions is a family of functions from $[n]$ to $[k]$ such that for
every $S \subseteq [n]$, $|S|=k$, there exists a function in the family that is
1-1 on $S$. There is an extensive literature dealing with explicit
constructions of perfect hash functions. The construction described
in \cite{JACM::AlonYZ1995} (following \cite{JACM::FredmanKS1984} and
\cite{SICOMP::SchmidtS1990}) is of size $2^{O(k)} \log n$. The best
known explicit construction is of size $e^k k^{O(\log k)} \log n$,
which closely matches the known lower bound of $\Omega(e^k \log n /
\sqrt{k})$ \cite{FOCS::NaorSS1995}.

\textbf{Finding and Counting Paths and Cycles.} The foundations for
the graph algorithms presented in this paper have been laid in
\cite{JACM::AlonYZ1995}. Two main randomized algorithms are
presented there, as follows. A simple directed or undirected path of
length $k-1$ in a graph $G=(V,E)$ that contains such a path can be
found in $2^{O(k)}|E|$ expected time in the directed case and in
$2^{O(k)}|V|$ expected time in the undirected case. A simple
directed or undirected cycle of size $k$ in a graph $G=(V,E)$ that
contains such a cycle can be found in either $2^{O(k)} |V||E|$ or
$2^{O(k)}|V|^\omega$ expected time, where $\omega < 2.376$ is the
exponent of matrix multiplication. The derandomization of these
algorithms incur an extra $\log |V|$ factor. As for the case of even
cycles, it is shown in \cite{Yuster:1997:FEC} that for every fixed
$k \geq 2$, there is an $O(|V|^2)$ algorithm for finding a simple
cycle of size $2k$ in an undirected graph. Improved algorithms for
detecting given length cycles have been presented in
\cite{ALGOR::AlonYZ1997} and \cite{conf/soda/YusterZ04}. An
interesting result from \cite{ALGOR::AlonYZ1997}, related to the
questions addressed in the present
paper, is an $O(|V|^\omega)$ algorithm for counting the number of
cycles of size at most $7$. Flum and Grohe proved that the problem of
counting 
\emph{exactly} the number of paths and cycles of length $k$ in both
directed and undirected graphs, parameterized by $k$, is
$\#W[1]$-complete \cite{Flum:2004:PCC}. Their result implies that most
likely there is no $f(k) \cdot n^c$-algorithm for counting the precise number 
of paths or
cycles of length $k$ in a graph of size $n$ for any computable
function $f: \bbbn \to \bbbn$ and constant $c$. This suggests the problem of
approximating these quantities.
Arvind and
Raman obtained a \emph{randomized} fixed-parameter tractable
algorithm to approximately count the number of copies of a fixed
subgraph with bounded treewidth within a large graph
\cite{conf/isaac/ArvindR02}. We settle in the affirmative the open
question they raise concerning the existence of a
\emph{deterministic} approximate counting algorithm for this
problem. For simplicity, we give algorithms for approximately
counting paths and cycles. These results can be easily extended to the
problem of approximately counting bounded treewidth subgraphs, combining the same
approach with the
method of \cite{JACM::AlonYZ1995}. The main new ingredient in our deterministic 
algorithms is the application of balanced families of perfect hash functions-
a combinatorial notion introduced here which, while simple, appears to be very useful.

\textbf{Balanced Families of Perfect Hash Functions.} We say that
a family of functions from $[n]$ to $[k]$ is a $\delta$-balanced
$(n,k)$-family of perfect hash functions if for every $S \subseteq
[n]$, $|S|=k$, the number of functions that are 1-1 on $S$ is
between $T/\delta$ and $\delta T$ for some constant $T>0$.
Balanced families of perfect hash functions are a natural
generalization of the usual concept of perfect hash functions. To
assist with our explicit constructions, we define also the even
more generalized notion of balanced splitters. (See section
\ref{sec:composing} for the definition. This is a generalization
of an ordinary splitter defined in \cite{FOCS::NaorSS1995}.) 

\textbf{Our Results.} The main focus of the paper is on explicit
constructions of balanced families of perfect hash functions and
their applications. First, we give non-constructive upper bounds
on the size of different types of balanced splitters. Then, we
compare these bounds with those achieved by constructive
algorithms. Our main result is an explicit construction, for every
$1 < \delta \leq 2$, of a $\delta$-balanced $(n,k)$-family of
perfect hash functions of size $2^{O(k \log \log
k)}(\delta-1)^{-O(\log k)} \log n$. The running time of the procedure that 
provides the construction 
is $2^{O(k \log \log k)}(\delta-1)^{-O(\log k)} n \log n +
(\delta-1)^{-O(k / \log k)}$.

Constructions of balanced families of perfect hash functions can
be applied to various counting problems in graphs. In particular,
we describe
deterministic algorithms that approximate the number of times that
a small subgraph appears within a large graph. The approximation
is always up to some multiplicative factor, that can be made
arbitrarily close
to $1$. For any $1 < \delta \leq 2$, the number of simple paths of
length $k-1$ in a graph $G=(V,E)$ can be approximated up to a
multiplicative factor of $\delta$ in time $2^{O(k \log \log
k)}(\delta-1)^{-O(\log k)} |E| \log |V| + (\delta-1)^{-O(k / \log
k)}$. The number of simple cycles of size $k$ can be approximated
up to a multiplicative factor of $\delta$ in time $2^{O(k \log
\log k)}(\delta-1)^{-O(\log k)} |E| |V| \log |V| +
(\delta-1)^{-O(k / \log k)}$.

\textbf{Techniques.} We use probabilistic arguments in order to
prove the existence of different types of small size balanced splitters (whose
precise definition is given in the next section).
To construct a balanced splitter, a natural randomized
algorithm is to choose a large enough number of independent random
functions. We show that in some cases, the method of conditional
probabilities, when applied on a proper choice of a potential
function, can derandomize this process in an efficient way.
Constructions of small probability spaces that admit $k$-wise
independent random variables are also a natural tool for achieving
good splitting properties. The use of error correcting codes is
shown to be useful when we want to find a family of functions from
$[n]$ to $[l]$, where $l$ is much bigger than $k^2$, such that for
every $S \subseteq [n]$, $|S|=k$, almost all of the functions
should be 1-1 on $S$. Balanced splitters can be composed in
different ways and our main construction is achieved by composing
three types of splitters. We apply the explicit constructions of
balanced families of perfect hash functions together with the
color-coding technique to get our approximate counting algorithms.

\section{Balanced Families of Perfect Hash Functions}\label{sec:composing}

In this section we formally define the new notions of balanced
families of perfect hash functions and balanced splitters. Here
are a few basics first. Denote by $[n]$ the set $\{1, \ldots
,n\}$. For any $k$, $1 \leq k \leq n$, the family of $k$-sized
subsets of $[n]$ is denoted by $\binom{[n]}{k}$. We denote by $k \
mod \ l$ the unique integer $0 \leq r < l$ for which $k = ql+r$,
for some integer $q$. We now introduce the new notion of balanced
families of perfect hash functions.

\begin{definition}
Suppose that $1 \leq k \leq n$ and $\delta \geq 1$. We say that a
family of functions from $[n]$ to $[k]$ is a $\delta$-balanced
$(n,k)$-family of perfect hash functions if there exists a
constant real number $T > 0$, such that for every $S \in
\binom{[n]}{k}$, the number of functions that are 1-1 on $S$,
which we denote by $inj(S)$, satisfies the relation $T / \delta
\leq inj(S) \leq \delta T$.
\end{definition}

The following definition generalizes both the last definition and
the definition of a splitter from \cite{FOCS::NaorSS1995}.

\begin{definition}
Suppose that $1 \leq k \leq n$ and $\delta \geq 1$, and let $H$ be
a family of functions from $[n]$ to $[l]$. For a set $S \in
\binom{[n]}{k}$ we denote by $split(S)$ the number of functions $h
\in H$ that split $S$ into equal-sized parts $h^{-1}(j) \bigcap
S$, $j=1, \ldots ,l$. In case $l$ does not divide $k$ we separate
between two cases. If $k \leq l$, then $split(S)$ is defined to be
the number of functions that are 1-1 on $S$. Otherwise, $k > l$
and we require the first $k \ mod \ l$ parts to be of size $\lceil
k/l \rceil$ and the remaining parts to be of size $\lfloor k/l
\rfloor$. We say that $H$ is a $\delta$-balanced
$(n,k,l)$-splitter if there exists a constant real number $T > 0$,
such that for every $S \in \binom{[n]}{k}$ we have $T / \delta
\leq split(S) \leq \delta T$.
\end{definition}

The definitions of balanced families of perfect hash functions and
balanced splitters given above enable us to state the following easy
composition lemmas.

\begin{lemma}\label{composition1}
For any $k < l$, let $H$ be an explicit $\delta$-balanced
$(n,k,l)$-splitter of size $N$ and let $G$ be an explicit
$\gamma$-balanced $(l,k)$-family of perfect hash functions of size
$M$. We can use $H$ and $G$ to get an explicit $\delta
\gamma$-balanced $(n,k)$-family of perfect hash functions of size
$N M$.
\end{lemma}

\begin{proof}
We compose every function of $H$ with every function of $G$ and get
the needed result.
\qed
\end{proof}

\begin{lemma}\label{composition2}
For any $k > l$, let $H$ be an explicit $\delta$-balanced
$(n,k,l)$-splitter of size $N$. For every $j$, $j=1, \ldots ,l$,
let $G_j$ be an explicit $\gamma_j$-balanced $(n,k_j)$-family of
perfect hash functions of size $M_j$, where $k_j=\lceil k/l
\rceil$ for every $j \leq k \ mod \ l$ and $k_j=\lfloor k/l
\rfloor$ otherwise. We can use these constructions to get an
explicit $(\delta \prod_{j=1}^l \gamma_j)$-balanced $(n,k)$-family
of perfect hash functions of size $N \prod_{j=1}^l M_j$.
\end{lemma}

\begin{proof}
We divide the set $[k]$ into $l$ disjoint intervals $I_1, \ldots
,I_l$, where the size of $I_j$ is $k_j$ for every $j=1, \ldots ,l$.
We think of $G_j$ as a family of functions from $[n]$ to $I_j$. For
every combination of $h \in H$ and $g_j \in G_j$, $j=1, \ldots ,l$,
we create a new function that maps an element $x \in [n]$ to
$g_{h(x)} (x)$.
\qed
\end{proof}

\section{Probabilistic Constructions}\label{sec:prob}

We will use the following two claims: a variant of the Chernoff
bound (c.f., e.g., \cite{MR1885388}) and Robbins' formula
\cite{MR0228020} (a tight version of Stirling's formula).

\begin{claim}\label{chernoff}
Let $Y$ be the sum of mutually independent indicator random
variables, $\mu = E[Y]$. For all $1 \leq \delta \leq 2$,
$$
Pr[ \frac{\mu}{\delta} \leq Y \leq \delta\mu ] > 1 - 2
e^{-(\delta-1)^2\mu/8}.
$$
\end{claim}

\begin{claim}\label{robbins}
For every integer $n \geq 1$,
$$
\sqrt{2 \pi} n^{n+1/2} e^{-n+1/(12n+1)} < n! < \sqrt{2 \pi}
n^{n+1/2} e^{-n+1/(12n)}.
$$
\end{claim}

Now we state the results for $\delta$-balanced $(n,k,l)$-splitters
of the three types: $k=l$, $k<l$ and $k>l$.

\begin{theorem}\label{prob1}
For any $1 < \delta \leq 2$, there exists a $\delta$-balanced
$(n,k)$-family of perfect hash functions of size $O(\frac{e^k
\sqrt{k} \log n}{(\delta - 1)^2})$.
\end{theorem}

\begin{proof}
(sketch)
Set $p=k!/k^k$ and $M=\lceil\frac{8(k \ln n + 1 )}{p (\delta - 1
)^2}\rceil$. We choose $M$ independent random functions. For a
specific set $S \in \binom{[n]}{k}$, the expected number of
functions that are 1-1 on $S$ is exactly $pM$. By the Chernoff
bound, the probability that for at least one set $S \in
\binom{[n]}{k}$, the number of functions that are 1-1 on $S$ will
not be as needed is at most
$$
\binom{n}{k} 2 e^{-(\delta-1)^2 pM/8} \leq 2 \binom{n}{k} e^{-(k \ln
n + 1 )} < 1.
$$
\qed
\end{proof}

\begin{theorem}\label{prob2}
For any $k < l$ and $1 < \delta \leq 2$, there exists a
$\delta$-balanced $(n,k,l)$-splitter of size $O(\frac{e^{k^2/l} k
\log n}{(\delta - 1)^2})$.
\end{theorem}

\begin{proof}
(sketch)
We set $p=\frac{l!}{(l-k)! l^k}$ and $M=\lceil\frac{8(k \ln n + 1
)}{p (\delta - 1 )^2}\rceil$. Using Robbins' formula, we get
$$
\frac{1}{p} \leq e^{k+1/12} (1-\frac{k}{l})^{l-k+1/2} \leq
e^{k+1/12} e^{-\frac{k}{l}(l-k+1/2)} = e^{\frac{k^2-k/2}{l}+1/12}.
$$
We choose $M$ independent random functions and proceed as in the
proof of Theorem \ref{prob1}.
\qed
\end{proof}

For the case $k > l$, the probabilistic arguments from
\cite{FOCS::NaorSS1995} can be generalized to prove existence of
balanced $(n,k,l)$-splitters. Here we focus on the special case of
balanced $(n,k,2)$-splitters, which will be of interest later.

\begin{theorem}\label{prob3}
For any $k \geq 2$ and $1 < \delta \leq 2$, there exists a
$\delta$-balanced $(n,k,2)$-splitter of size $O(\frac{k \sqrt{k}
\log n}{(\delta - 1)^2})$.
\end{theorem}

\begin{proof}
(sketch)
Set $M=\lceil\frac{8(k \ln n + 1 )}{p (\delta - 1 )^2}\rceil$,
where $p$ denotes the probability to get the needed split in a
random function. If follows easily from Robbins' formula that
$1/p=O(\sqrt{k})$. We choose $M$ independent random functions and
proceed as in the proof of Theorem \ref{prob1}.
\qed
\end{proof}

\section{Explicit Constructions}\label{sec:explicit}

In this paper, we use the term explicit construction for an
algorithm that lists all the elements of the required family of
functions in time which is polynomial in the total size of the
functions. For a discussion on other definitions for this term,
the reader is referred to \cite{FOCS::NaorSS1995}. We state our
results for $\delta$-balanced $(n,k,l)$-splitters of the three
types: $k=l$, $k<l$ and $k>l$.

\begin{theorem}\label{explicit1}
For any $1 < \delta \leq 2$, a $\delta$-balanced $(n,k)$-family of
perfect hash functions of size $O(\frac{e^k \sqrt{k} \log
n}{(\delta - 1)^2})$ can be constructed deterministically within time
$\binom{n}{k}\frac{e^k k^{O(1)} n \log n}{(\delta - 1)^2}$.
\end{theorem}

\begin{proof}
We set $p=k!/k^k$ and $M=\lceil\frac{16(k \ln n + 1 )}{p (\delta -
1 )^2 }\rceil$. Denote $\lambda = (\delta - 1)/4$, so obviously $0
< \lambda \leq 1/4$. Consider a choice of $M$ independent random
functions from $[n]$ to $[k]$. This choice will be derandomized in
the course of the algorithm. For every $S \in \binom{[n]}{k}$, we
define $X_S=\sum_{i=1}^M X_{S,i}$, where $X_{S,i}$ is the
indicator random variable that is equal to $1$ iff the $i$th
function is 1-1 on $S$. Consider the following potential function:
$$
\Phi = \sum_{S \in \binom{[n]}{k}} e^{\lambda(X_S-pM)}+e^{\lambda(pM
- X_S)}.
$$
Its expectation can be calculated as follows:
$$
E[\Phi] = \binom{n}{k} ( e^{-\lambda pM} \prod_{i=1}^ME[e^{\lambda
X_{S,i}}] + e^{\lambda pM} \prod_{i=1}^ME[e^{-\lambda X_{S,i}}]) =
$$
$$
=\binom{n}{k} ( e^{-\lambda pM} [p e^{\lambda} + (1-p) ]^M +
e^{\lambda pM} [p e^{-\lambda} + (1-p) ]^M).
$$

We now give an upper bound for $E[\Phi]$. Since $1+u \leq e^u$ for
all $u$ and $e^{-u} \leq 1-u+u^2/2$ for all $u \geq 0$, we get
that $p e^{-\lambda} + (1-p) \leq e^{p(e^{-\lambda}-1)} \leq
e^{p(-\lambda+\lambda^2/2)}$. Define $\epsilon = e^{\lambda}-1$,
that is $\lambda=\ln(1+\epsilon)$. Thus $p e^{\lambda} + (1-p) =
1+\epsilon p \leq e^{\epsilon p}$. This implies that
$$
E[\Phi] \leq n^k ( (\frac{e^{\epsilon}}{1+\epsilon})^{pM} +
e^{\lambda^2 pM/2} ).
$$
Since $e^u \leq 1+u+u^2$ for all $0 \leq u \leq 1$, we have that
$\frac{e^{\epsilon}}{1+\epsilon} = e^{e^{\lambda}-1-\lambda} \leq
e^{\lambda^2}$. We conclude that
$$
E[\Phi] \leq 2 n^k e^{\lambda^2 pM} \leq e^{2(k \ln n + 1)}.
$$

We now describe a deterministic algorithm for finding $M$
functions, so that $E[\Phi]$ will still obey the last upper bound.
This is performed using the method of conditional probabilities
(c.f., e.g., \cite{MR1885388}, chapter 15). The algorithm will
have $M$ phases, where each phase will consist of $n$ steps. In
step $i$ of phase $j$ the algorithm will determine the $i$th value
of the $j$th function. Out of the $k$ possible values, we greedily
choose the value that will decrease $E[\Phi]$ as much as possible.
We note that at any specific step of the algorithm, the exact
value of the conditional expectation of the potential function can
be easily computed in time $\binom{n}{k} k^{O(1)}$.

After all the $M$ functions have been determined, every set $S \in
\binom{[n]}{k}$ satisfies the following:
$$
e^{\lambda(X_S-pM)}+e^{\lambda(pM - X_S)} \leq e^{2(k \ln n + 1)}.
$$
This implies that
$$
-2(k \ln n + 1 )\leq \lambda(X_S-pM) \leq 2(k \ln n + 1).
$$
Recall that $\lambda = (\delta - 1)/4$, and therefore
$$
(1-\frac{8(k \ln n + 1)}{(\delta - 1)pM}) pM \leq X_S \leq
(1+\frac{8(k \ln n + 1)}{(\delta - 1)pM}) pM.
$$
Plugging in the values of $M$ and $p$ we get that
$$
(1-\frac{\delta-1}{2})pM \leq X_S \leq (1+\frac{\delta-1}{2})pM.
$$
Using the fact that $1/u \leq 1-(u-1)/2$ for all $1 \leq u \leq 2$,
we get the desired result
$$
pM/\delta \leq X_S \leq \delta pM.
$$
\qed
\end{proof}

\begin{theorem}\label{explicit2}
For any $1 < \delta \leq 2$, a $\delta$-balanced $(n,k,\lceil
\frac{2 k^2}{\delta-1} \rceil)$-splitter of size $\frac{k^{O(1)}
\log n}{(\delta - 1)^{O(1)}}$ can be constructed in time
$\frac{k^{O(1)} n \log n}{(\delta - 1)^{O(1)}}$.
\end{theorem}

\begin{proof}
Denote $q=\lceil \frac{2 k^2}{\delta-1} \rceil$. Consider an
explicit construction of an error correcting code with $n$ codewords
over alphabet $[q]$ whose normalized Hamming distance is at least $1
- \frac{2}{q}$. Such explicit codes of length $O(q^2 \log n)$ exist
\cite{journals/tit/AlonBNNR92}. Now let every index of the code
corresponds to a function from $[n]$ to $[q]$. If we denote by $M$
the length of the code, which is in fact the size of the splitter,
then for every $S \in \binom{[n]}{k}$, the number of good splits is
at least
$$
(1-\binom{k}{2} \frac{2}{q}) M \geq (1-\frac{\delta-1}{2}) M \geq M
/\delta,
$$
where the last inequality follows from the fact that $1-(u-1)/2
\geq 1/u$ for all $1 \leq u \leq 2$.
\qed
\end{proof}

For our next construction we use small probability spaces that
support a sequence of almost $k$-size independent random variables.
A sequence $X_1, \ldots ,X_n$ of random Boolean variables is
$(\epsilon,k)$-independent if for any $k$ positions $i_1 < \cdots <
i_k$ and any $k$ bits $\alpha_1, \ldots ,\alpha_k$ we have
$$
|Pr[X_{i_1}=\alpha_1, \ldots ,X_{i_k}=\alpha_k]-2^{-k}| < \epsilon.
$$
It is known
(\cite{SICOMP::NaorN1993},\cite{journals/rsa/AlonGHP92},\cite{journals/tit/AlonBNNR92})
that sample spaces of size $2^{O(k+\log \frac{1}{\epsilon})} \log
n$ that support $n$ random variables that are
$(\epsilon,k)$-independent can be constructed in time $2^{O(k+\log
\frac{1}{\epsilon})} n \log n$.

\begin{theorem}\label{explicit3}
For any $k \geq l$ and $1 < \delta \leq 2$, a $\delta$-balanced
$(n,k,l)$-splitter of size $2^{O(k \log l - \log(\delta-1))} \log
n$ can be constructed in time $2^{O(k \log l - \log(\delta-1))} n
\log n$.
\end{theorem}

\begin{proof}
We use an explicit probability space of size $2^{O(k \log l -
\log(\delta-1))} \log n$ that supports $n \lceil \log_2 l \rceil$
random variables that are $(\epsilon,k \lceil \log_2 l
\rceil)$-independent where $\epsilon = 2^{-k \lceil \log_2 l
\rceil - 1} (\delta - 1)$. We attach $\lceil \log_2 l \rceil$
random variables to each element of $[n]$, thereby assigning it a
value from $[2^{\lceil \log_2 l \rceil}]$. In case $l$ is not a
power of $2$, all elements of $[2^{\lceil \log_2 l \rceil}]-[l]$
can be mapped to $[l]$ by some arbitrary fixed function. If
follows from the construction that there exists a constant $T > 0$
so that for every $S \in \binom{[n]}{k}$, the number of good
splits satisfies
$$
\frac{T}{\delta} \leq (1-\frac{\delta-1}{2}) T \leq split(S) \leq
(1+\frac{\delta-1}{2}) T \leq \delta T.
$$
\qed
\end{proof}

\begin{corollary}
For any fixed $c > 0$, a $(1+c^{-k})$-balanced $(n,k,2)$-splitter
of size $2^{O(k)} \log n$ can be constructed in time $2^{O(k)} n
\log n$.
\end{corollary}

Setting $l=k$ in Theorem \ref{explicit3}, we get that a
$\delta$-balanced $(n,k)$-family of perfect hash functions of size
$2^{O(k \log k - \log(\delta-1))} \log n$ can be constructed in
time $2^{O(k \log k - \log(\delta-1))} n \log n$. Note that if $k$
is small enough with respect to $n$, say $k=O(\log n/\log \log n)$, then for any fixed
$1 < \delta \leq 2$, this already gives a family of functions of
size polynomial in $n$. We improve upon this last result in the
following Theorem, which is our main construction.

\begin{theorem}\label{explicit4}
For $1 < \delta \leq 2$, a $\delta$-balanced $(n,k)$-family of
perfect hash functions of size $\frac{2^{O(k \log \log
k)}}{(\delta-1)^{O(\log k)}} \log n$ can be constructed in time
$\frac{2^{O(k \log \log k)}}{(\delta-1)^{O(\log k)}} n \log n +
(\delta-1)^{-O(k / \log k)}$. In particular, for any fixed $1 <
\delta \leq 2$, the size is $2^{O(k \log \log k)} \log n$ and the
time is $2^{O(k \log \log k)} n \log n$.
\end{theorem}

\begin{proof}
(sketch)
Denote $l=\lceil \log_2 k \rceil$ ,$\delta'=\delta^{1/3}$,
$\delta''=\delta^{1/(3l)}$, and $q=\lceil \frac{2 k^2}{\delta'-1}
\rceil$. Let $H$ be a $\delta'$-balanced $(q,k,l)$-splitter of
size $2^{O(k \log \log k)} (\delta' - 1)^{-O(1)}$ constructed
using Theorem \ref{explicit3}. For every $j$, $j=1, \ldots ,l$,
let $B_j$ be a $\delta''$-balanced $(q,k_j)$-family of perfect
hash functions of size $O(e^{k/\log k} k)(\delta''-1)^{-O(1)}$
constructed using Theorem \ref{explicit1}, where $k_j=\lceil k/l
\rceil$ for every $j \leq k \ mod \ l$ and $k_j=\lfloor k/l
\rfloor$ otherwise. Using Lemma \ref{composition2} for composing
$H$ and $\{B_j\}_{j=1}^l$, we get a $\delta'^2$-balanced
$(q,k)$-family $D'$ of perfect hash functions.

Now let $D''$ be a $\delta'$-balanced $(n,k,q)$-splitter of size
$k^{O(1)} (\delta'-1)^{-O(1)} \log n$ constructed using Theorem
\ref{explicit2}. Using Lemma \ref{composition1} for composing $D'$
and $D''$, we get a $\delta$-balanced $(n,k)$-family of perfect
hash functions, as needed. Note that for calculating the size of
each $B_j$, we use the fact that $e^{u/2} \leq 1+u \leq e^u$ for
all $0 \leq u \leq 1$, and get the following:
$$
\delta''-1 = (1+(\delta-1))^{\frac{1}{3l}}-1 \geq e^{\frac{\delta -
1}{6l}}-1 \geq \frac{\delta - 1}{6l}.
$$
The time needed to construct each $B_j$ is $2^{O(k)}
(\delta'-1)^{-O(k / \log k)}$. The $2^{O(k)}$ term is omitted in the
final result, as it is negligible in respect to the other terms.

\qed
\end{proof}

\section{Approximate Counting of Paths and Cycles}\label{sec:counting}

We now state what it means for an algorithm to approximate a
counting problem.

\begin{definition}
We say that an algorithms approximates a counting problem by a
multiplicative factor $\delta \geq 1$ if for every input $x$, the
output $ALG(x)$ of the algorithm satisfies $N(x)/\delta \leq ALG(x)
\leq \delta N(x)$, where $N(x)$ is the exact output of the counting
problem for input $x$.
\end{definition}

The technique of color-coding is used for approximate counting of
paths and cycles. Let $G=(V,E)$ be a directed or undirected graph.
In our algorithms we will use constructions of balanced
$(|V|,k)$-families of perfect hash functions. Each such function
defines a coloring of the vertices of the graph. A path is said to
be \textit{colorful} if each vertex on it is colored by a distinct
color. Our goal is to count the exact number of colorful paths in
each of these colorings.

\begin{theorem}\label{application1}
For any $1 < \delta \leq 2$, the number of simple (directed or
undirected) paths of length $k-1$ in a (directed or undirected)
graph $G=(V,E)$ can be approximated up to a multiplicative factor
of $\delta$ in time $\frac{2^{O(k \log \log
k)}}{(\delta-1)^{O(\log k)}} |E| \log |V| + (\delta-1)^{-O(k /
\log k)}$.
\end{theorem}

\begin{proof}
(sketch)
We use the $\delta$-balanced $(|V|,k)$-family of perfect hash
functions constructed using Theorem \ref{explicit4}. Each function
of the family defines a coloring of the vertices in $k$ colors. We
know that there exists a constant $T>0$, so that for each set $S
\subseteq V$ of $k$ vertices, the number of functions that are 1-1
on $S$ is between $T/\delta$ and $\delta T$. The exact value of
$T$ can be easily calculated in all of our explicit constructions.

For each coloring, we use a dynamic programming approach in order
to calculate the exact number of colorful paths. We do this in $k$
phases. In the $i$th phase, for each vertex $v \in V$ and for each
subset $C \subseteq \{1, \ldots ,k\}$ of $i$ colors, we calculate
the number of colorful paths of length $i-1$ that end at $v$ and use
the colors of $C$. To do so, for every edge $(u,v) \in E$, we
check whether it can be the last edge of a colorful path of length
$i-1$ ending at either $u$ or $v$. Its contribution to the number of
paths of length $i-1$ is calculated using our knowledge on the
number of paths of length $i-2$. The initialization of phase $1$
is easy and after performing phase $k$ we know the exact number
of paths of length $k-1$ that end at each vertex $v \in V$. The time
to process each coloring is therefore $2^{O(k)} |E|$.

We sum the results over all colorings and all ending vertices $v
\in V$. The result is divided by $T$. In case the graph is
undirected ,we further divide by $2$. This is guaranteed to be the
needed approximation.
\qed
\end{proof}

\begin{theorem}\label{application2}
For any $1 < \delta \leq 2$, the number of simple (directed or
undirected) cycles of size $k$ in a (directed or undirected) graph
$G=(V,E)$ can be approximated up to a multiplicative factor of
$\delta$ in time $\frac{2^{O(k \log \log k)}}{(\delta-1)^{O(\log
k)}} |E| |V| \log |V| + (\delta-1)^{-O(k / \log k)}$.
\end{theorem}

\begin{proof}
(sketch)
We use the $\delta$-balanced $(|V|,k)$-family of perfect hash
functions constructed using Theorem \ref{explicit4}. For every set
$S$ of $k$ vertices, the number of functions that are 1-1 on $S$
is between $T/\delta$ and $\delta T$. Every function defines a
coloring and for each such coloring we proceed as follows. For
every vertex $s \in V$ we run the algorithm described in the proof
of Theorem \ref{application1} in order to calculate for each
vertex $v \in V$ the exact number of colorful paths of length
$k-1$ from $s$ to $v$. In case there is an edge $(v,s)$ that
completes a cycle, we add the result to our count.

We sum the results over all the colorings and all pairs of
vertices $s$ and $v$ as described above. The result is divided by
$kT$. In case the graph is undirected, we further divide by $2$.
The needed approximation is achieved.
\qed
\end{proof}

\begin{corollary}
For any constant $c>0$, there is a deterministic polynomial time
algorithm for approximating both the number of simple paths of
length $k$ and the number of simple cycles of size $k$ for every
$k \leq O(\frac{\log n}{\log \log \log n})$ in a graph with $n$
vertices, where the approximation is up to a multiplicative factor
of $1+(\ln \ln n)^{-c \ln \ln n}$.
\end{corollary}

\section{Concluding Remarks}\label{sec:conclude}

\begin{itemize}
\item An interesting open problem is whether for every fixed
$\delta>1$, there exists an explicit $\delta$-balanced
$(n,k)$-family of perfect hash functions of size $2^{O(k)} \log
n$. The key ingredient needed is an improved construction of
balanced $(n,k,2)$-splitters. Such splitters can be applied
successively to get the balanced $(n,k,\lceil \log_2 k
\rceil)$-splitter needed in Theorem \ref{explicit4}. It seems that
the constructions presented in \cite{journals/rsa/AlonGHP92} could
be good candidates for balanced $(n,k,2)$-splitters, although the
Fourier analysis in this case (along the lines of
\cite{journals/combinatorica/AzarMN98}) seems elusive.

\item Other algorithms from \cite{JACM::AlonYZ1995} can be
generalized to deal with counting problems. In particular it is
possible to combine our approach here with the ideas of
\cite{JACM::AlonYZ1995} based on fast matrix multiplication in
order to approximate the number of cycles of a given length. Given
a forest $F$ on $k$ vertices, the number of subgraphs of $G$
isomorphic to $F$ can be approximated using a recursive algorithm
similar to the one in \cite{JACM::AlonYZ1995}. For a weighted
graph, we can approximate, for example, both the number of minimum
(maximum) weight paths of length $k-1$ and the number of minimum
(maximum) weight cycles of size $k$. Finally, all the results can be readily extended
from paths and cycles to arbitrary small subgraphs of bounded tree-width.
We omit the details.

\item
In the definition of a balanced $(n,k)$-family of perfect hash
functions, there is some constant $T>0$, such that for every $S
\subseteq [n]$, $|S|=k$, the number of functions that are 1-1 on $S$
is close to $T$. We note that the value of $T$ need not be equal to
the expected number of 1-1 functions on a set of size $k$, for the case
that the functions were chosen independently according to a uniform
distribution. For example, the value of $T$ in the construction of
Theorem \ref{explicit4} is not even asymptotically equal to what one
would expect in a uniform distribution.
\end{itemize}


\end{document}